\documentclass{llncs}

\usepackage{fullpage}
\usepackage{makeidx}  
\usepackage{amsmath}
\usepackage{algorithm}

\numberwithin{algorithm}{section}
\begin{document}
\pagestyle{headings}  
\title{On the Problem of Wireless Scheduling with Linear Power Levels}
\author{\textsc{Tigran Tonoyan}\thanks{Research partially founded by FRONTS 215270}}

\institute{TCS Sensor Lab,
Centre Universitaire d'Informatique, route de Drize 7, 1227 Carouge, Geneva, Switzerland}

\maketitle              

\begin{abstract}
In this paper we consider the problem of communication scheduling in wireless networks with respect to the SINR(Signal to Interference plus Noise Ratio) constraint in metric spaces. The nodes are assigned \emph{linear powers}, i.e. for each sender node the power is constant times the path loss between the sender and corresponding receiver. This is the minimal power for a successful transmission. We present a constant factor deterministic approximation algorithm, which works for at least Euclidean \emph{fading metrics}. Simultaneously we obtain the approximate value of the optimal schedule length with error at most a constant factor. To give an insight into the complexity of the problem, we show that in some metric spaces the problem is NP-hard and cannot be approximated within a factor less than 1.5. 
\end{abstract}

	\section{Introduction}
The problem of scheduling is the following: given some set of transmission requests or \emph{links} (sender-receiver node pairs in the network), the goal is to find a schedule, so that all the transmissions between those nodes can be done successfully in the minimum time. The main factor affecting the successful data transmissions in wireless networks is the signal interference of the concurrently transmitting nodes, which in general makes it impossible to do all the needed transmissions concurrently: there can be a receiver node, which cannot decode the data intended to it because of the ``noise'' made by other transmitions. So one needs to split the set of requests into subgroups, in each of which all sender nodes can transmit concurrently. Then all the data transmission can be done in a time proportional to the number of different groups in the schedule. The goal is to minimize the number of such subsets.

The solution of the problem depends crucially on the model of interference which is adopted for the given case. There are several models considered in the literature, such as \emph{protocol model} and \emph{physical model}. We consider the \emph{physical model}, which is shown to be more realistic than (traditional) protocol model. It assumes that the influence of a transmitting node on other nodes decreases proportionally to a constant power of the distance from that node (if there are no obstacles). Based on this model, the \emph{SINR (Signal to Interference plus Noise Ratio)} constraint is considered for reflecting the possibility or impossibility of accepting the signal of some sender by the corresponding receiver.

The solution of the scheduling problem depends also on the power levels of the nodes in the network: each node can transmit the data with a specific power, so the more is the power of a node, the stronger is the signal received by the intended receiver (also the more is the ``noise'' made by that node to other transmissions). Our results are for \emph{linear} power assignments, which along with \emph{uniform} power assignments are the most popular power schemes used in the literature. In case of the uniform power assignment all the nodes use the same power level, so this assignment is very simple to implement, but can lead to long schedules (hence an increased delay in the network). In case of linear power assignment the power of the sender of a link is constant times the \emph{path loss} between the sender and corresponding receiver, so this is the minimum possible power for a successful transmission. On the other hand, the optimum schedule length of the linear power can also be too long compared to some other power assignments.

	\textit{Related Work and Our Results}
The algorithmic study of the problem of scheduling for arbitrary networks in SINR model gained considerable attention in last years because of the attempts to show that this model describes the physical reality more precise than more traditional \emph{graph-based models}, such as the \emph{protocol model}~\cite{antigrp}.  

There are several variants of scheduling problem considered in the literature. In~\cite{connectivity},~\cite{connectivity1}, \emph{connectivity} problem is considered from the scheduling perspective, where it is needed to find and schedule (into minimal number of slots) a set of links which form a connected structure. \emph{Scheduling with power control}, where for obtaining small schedules it is allowed as well to control the power levels of the transmitters, is considered in~\cite{first},~\cite{oblivious},~\cite{hal1},~\cite{halldorsson1},~\cite{tonoyanpc},~\cite{kesselheimpc}. In~\cite{kesselheimpc} a $\log{n}$-approximation algorithm is designed (where $n$ is the number of links), which uses power assignments of \emph{non-local} nature. In~\cite{hal1},~\cite{halldorsson1},~\cite{tonoyanpc} it is shown that the \emph{mean power assignment} is relatively efficient from the point of view of scheduling, when compared to other local power assignments, and approximation guarantees are proven for this power assignment. In fact in~\cite{oblivious}, generalizing the construction from~\cite{connectivity}, it is shown that for each local (or \emph{oblivious}, as they call it) power assignment, there are network instances, for which no non-trivial schedules can be obtained using this power assignment. On the other hand, in~\cite{hal1} it is shown that if the lengths of links differ not more than a constant factor, then the \emph{uniform power assignment} is a constant factor approximation for this problem.

The problem of scheduling with fixed power levels is consider for several power assignments, such as \emph{uniform, linear and mean} power assignments. There are $O(\log{n})$ approximation algorithms designed for the mean power assignment in~\cite{tonoyanpc}. The case of uniform power assignment is considered in~\cite{santienq},~\cite{gus2},~\cite{gus1},~\cite{dinitz},~\cite{hal}. To the best of our knowledge, the best approximation ratio obtained for this case is $O(\log{n})$, although constant factor approximation algorithms are designed for a related problem of \emph{capacity maximization} for a large family of power assignments~\cite{halmit} (in~\cite{hal} it is claimed that their algorithm approximates the optimal schedule length within a constant factor, but a flow was found in the proof). In~\cite{kesvok} $O(\log^2{n})$-approximation randomized distributed algorithms are designed for a range of power assignments, which is improved to $O(\log{n})$ factor in~\cite{halmit1}. The specific case of linear power assignments is considered e.g. in~\cite{chafekar},~\cite{fng}. In~\cite{fng} a randomized algorithm is proposed, which finds a schedule of length $O(I+\log^2{n})$, where $I$ is a lower bound on the optimal schedule length.

The paper is an extended version of~\cite{tonoyanlin}, where we propose a constant factor algorithm for scheduling w.r.t. fixed linear power assignment, and show that the optimal schedule length is $\Theta(I)$, where $I$ is as defined in~\cite{fng}. We also show that the problem cannot be approximated within a factor less than 1.5 (with assumption $P\neq NP$). 
 Whith the same methods as in ~\cite{fng}, our algorithm for linear powers also leads to a randomized algorithm for Cross-Layer Optimization problem, which improves the algorithms from ~\cite{fng} by a factor $\log{n}$.

	\section{ Formal Definition of the Problem} \label{notations}
Throughout this work we assume the wireless network nodes to be statically located (i.e. the network is not mobile) in a metric measure space $X$ with a distance function $d$ and a measure $\mu$.

The ball in $X$ with center $p$ and radius $r>0$ is the set 
\[
B(p, r)=\left\{q\in X| d(p, q)<r \right\},
\]
and the ring with center $p$, width $w>0$ and outer radius $r>w$ is the set
\[
R(p, r, w)=\left\{q\in X|r-w\leq d(p,q)<r\right\}.
\]
For a ring $R=R(p,r,w)$ we denote $b(R)=B(p, r)$ and $B(R)=B(p, r-w)$.

We assume that the measure $\mu$ satisfies the following condition: for any two balls $A$ and $B$ with radii $a$ and $b$ respectively,
\begin{equation}\label{E:metric}
\frac{\mu(A)}{\mu(B)}\leq K\left(\frac{a}{b}\right)^m
\end{equation}
holds for some constants $K\geq1$ and $m\geq1$, which are specific to the metric space.

We are given a set of links $L=\{1,2,\dots,n\}$, where each link $v$ represents a communication request from a sender node $s_v$ to a receiver node $r_v$. The \emph{asymmetric distance} from a link $v$ to a link $w$ is $d_{vw}=d\left(s_v,r_w\right)$. The length of the link $v$ is  $d_{vv}=d\left(s_v,r_v\right)$. Each transmitter $s_v$ is assigned a power level $P_v$, which does not change. We assume that the strength of the signal decreases with the distance from the transmitter, i.e. the received signal strength from the sender of $w$ at the receiver of $v$ is $\displaystyle P_{wv}=\frac{P_w}{d_{wv}^\alpha}$, where $\alpha>0$ denotes the \emph{path-loss exponent}. For interference we adopt the \emph{SINR} model, where the transmission corresponding to a link $v$ is successful if and only if the following condition holds:
\begin{equation}\label{E:SINR}
P_{vv}\geq\beta\left(\sum_{w \in S\setminus v}P_{wv}+N\right),
\end{equation}
where $N\geq 0$ denotes the ambient noise, $\beta>1$ denotes the minimum SINR required for message to be successfully received, and $S$ is the set of concurrently scheduled links. We say that $S$ is \emph{feasible} if~(\ref{E:SINR}) is satisfied for all $v\in S$. The linear power assignment assigns each sender $s_v$ a power level $P_v=c_ld_{vv}^\alpha$, where $c_l$. The uniform power assignment assigns to each sender node the same power level $P$. 

A partition of the set $L$ into feasible subsets (or \emph{slots}) is called \emph{a schedule}. The number of subsets in a schedule is called \emph{the length} of the schedule. The problem we are interested in is to find a schedule of a minimum length, assuming that the linear power assignment is used.

	\section{Auxiliary Facts}\label{auxiliary}

Here is a set of lemmas, which we will use in subsequent sections.

The follwing is a known bound for Riemann zeta function:
\begin{equation}\label{E:series}
\zeta(s)=\sum_{i=1}^{\infty}\frac{1}{i^s}\leq\frac{s}{s-1}\mbox{ , if }s>1,
\end{equation}
which can be proven by noticing that $\displaystyle\sum_{i=1}^{\infty}\frac{1}{i^s}\leq \int_{1}^{\infty}{\frac{1}{x^s}dx}+1$.

A proof of the following lemma can be found, for example, in~\cite{hrd}, page 28.
\begin{lemma}\label{L:ineq}
For real numbers $a_1,a_2,\dots,a_m(a_i\geq 0,i=1,2,\dots,m)$, and $r,s(0<r<s)$,
\[
\left(\sum_{i=1}^ma_i^s\right)^{\frac{1}{s}}<\left(\sum_{i=1}^ma_i^r\right)^{\frac{1}{r}}
\] holds, unless all $a_i$ but one are zero.
\end{lemma}

For the next lemma, consider any given real numbers $a\geq 1$ and $c>0$, and the function $f(t)=(a+c)^t-a^t$. Note that $f(t)$ is a monotonically increasing function on $[1,\infty]$, as $f'(t)>0$ for $t\geq 1$. So $f(t)\leq f(\lceil t\rceil)$ for $t\geq 1$. For an integer $k\geq1$ we have
\[
(a+c)^k-a^k=c\sum_{i=0}^{k-1}{(a+c)^ia^{k-1-i}}\leq kc(a+c)^{k-1},
\]
so we have the following lemma:
\begin{lemma}\label{L:exponentialfunction}
For real numbers $a\geq 1,c>0,t\geq1$
\[
(a+c)^t-a^t\leq \lceil t\rceil c(a+c)^{\lceil t\rceil-1}.
\]
\end{lemma}

	\section{The scheduling algorithm} \label{algorithm}
In this section we present a scheduling algorithm, which is very simple and approximates the optimal schedule length within a constant factor. It is assumed that the linear power assignment is used for the power levels.

As in other scheduling algorithms, instead of using the SINR formula in the form of~(\ref{E:SINR}), we use the inverse of it, which has the useful property of being additive, and is easier to deal with.
\begin{definition}
The affectance of a link $v$, caused by a set $S$ of links, is the following sum of relative interferences of the links from $S$ on $v$,
\[
a_S(v)=\sum_{w \in S\setminus v} \left(\frac{d_{ww}}{d_{wv}}\right)^\alpha.
\]
\end{definition}
With the affectance defined, SINR constraint for a set of links $S$ and a link $v$ can be written as 
\begin{equation}\label{E:SINRlin}
\displaystyle a_S(v)\leq \frac{1}{\beta}-\frac{N}{c_l}.
\end{equation}
 For simplicity of writing we denote the right side by $1/\beta$.

\subsection{Formulation of the algorithm}
The algorithm (pseudocode is presented as Algorithm~\ref{A:LINEAR}) is a \emph{greedy algorithm}, which sorts all the links in descending order of the length, and starting from the first one, adds each link to the first slot, in which already scheduled links influence this one no more than a predefined constant. As we will see afterwards, this special ordering is needed only for feasibility of the resulting schedule, whereas the proof of approximation factor does not depend on this order.

The precise value of the constant $c>3$ used in the algorithm will be defined afterwards.

\renewcommand {\labelenumi}{\arabic{enumi}.}
\renewcommand{\labelenumii}{\arabic{enumi}.\arabic{enumii}}
\begin{algorithm}\label{A:LINEAR}\caption{Scheduling w.r.t. linear power assignment.}
\begin{enumerate}
\item {Input: the links $1,2,\dots,n$}
\item {sort the links in descending order of their lengths: $l_1,l_2,\dots,l_n$}
\item {$S_i\leftarrow \emptyset,i=1,2,\dots$}
\item {for $t\leftarrow 1 \mbox{ to }n$ do}
\begin{enumerate}
\item {find the smallest $i$, such that $\displaystyle a_{S_i}\left(l_t\right)\leq\frac{1}{c^\alpha}$}
\item {schdule $l_t$ with $S_i$: $S_i:=S_i\cup l_t$}
\end{enumerate}
\item {output: $\left(S_1,S_2,\dots\right)$}
\end{enumerate}
\end{algorithm}

\subsection{Correctness of the algorithm}

Consider the set of links $S$ assigned to the same slot by the algorithm, and $v\in S$. Let $S^-$ denote the subset of $S$, which contains the links shorter than $v$. It is enough to show that $a_{S^-}(v)$ is small for each slot $S$ and $v\in S$. To show this we will use a standard area argument.

We start with a simple lemma, which shows that if two links are scheduled in the same slot, then they should be spatially separated. Let the links $w$ and $v$ be assigned to the same slot by the algorithm, and $d=\max {\left\{d_{vv},d_{ww}\right\}}$.
\begin{lemma}\label{L:sizes}
For any two links $w$ and $v$, which are as above, the following holds:
\[
d_{vw}\geq(c-2)d\mbox{, }d_{wv}\geq(c-2)d \mbox{ and } d\left(s_v,s_w\right)\geq(c-3)d.
\]
\end{lemma}

Feasibility of the schedule is shown
\begin{lemma} \label{L:spl}
There exists a constant $c_0$, depending only on $m$, $K$ and $\alpha$, such that for the link $v$ and the set of links $S^-$ as above,
\[
a_{S^-}(v)\leq \frac{c_0}{(c-3)^\alpha},
\]
holds, if $\alpha>\displaystyle\frac{m}{m+1-\lceil m\rceil}$.
\end{lemma}
\begin{proof}
For simplicity, throughout this proof we denote $q=c-2$. Consider the partition of the metric space into concentric rings $R_i=R(r_v,(i+1)qd_{vv},qd_{vv})$ for $i=1,2,\dots$, and the ball $B(r_v,qd_{vv})$. From Lemma~\ref{L:sizes} and definition of $S^-$ it follows that there are no senders from $S^-$ inside $B(r_v,qd_{vv})$. Now for some $i>0$ consider the links from $S^-$ with senders inside $R_i$, and denote that set by $S^-_i$. For each link $w$ denote $\displaystyle \rho_w=\frac{(q-1)d_{ww}}{2}$. Then it follows from the last inequality of Lemma~ \ref{L:sizes} that for each such link $w$ the ball $B(s_w,\rho_w)$ doesn't intersect the corresponding ball of any other link. Further, all such balls are contained in the ring $R_i'=R(r_v,(i+1)qd_{vv}+\rho_v,qd_{vv}+2\rho_v)$. So from the countable additivity of $\mu$ it follows that
\[
\sum_{w\in S^-_i}\mu(B(s_w,\rho_w))\leq\mu(R_i')=\mu(B(R_i'))-\mu(b(R_i')) \mbox{ or, as }K\geq1,
\]
\begin{equation}\label{E:areaargument}
\sum_{w\in S^-_i}\frac{\mu(B(s_w,\rho_w))}{\mu(B(R_i'))}\leq 1-\frac{\mu(b(R_i'))}{\mu(B(R_i'))}\leq K-\frac{\mu(b(R_i'))}{\mu(B(R_i'))}.
\end{equation}
From (\ref{E:metric}) we have the following inequalities for each link $w$:
\[
\frac{\mu(b(R_i'))}{\mu(B(R_i'))}\leq K\left(\frac{iqd_{vv}-\rho_v}{(i+1)qd_{vv}+\rho_v}\right)^m\mbox{ and}
\]
\[
\frac{\mu(B(s_w,\rho_w))}{\mu(B(R_i'))}\geq \frac{1}{K}\left(\frac{\rho_w}{(i+1)qd_{vv}+\rho_v}\right)^m,
\]
which combined with (\ref{E:areaargument}) leads to the following:
\begin{multline}\label{E:pr}
\sum_{w\in S^-_i}\rho_w^m\leq K^2\left(\left((i+1)qd_{vv}+\rho_v\right)^m-\left(iqd_{vv}-\rho_v\right)^m\right)\leq\\
\leq K^2\left(qd_{vv}\right)^m\left( (i+3/2)^m-(i-1/2)^m\right)\leq \frac{3^{m-1}}{2^{m-2}}\lceil m\rceil K^2\left(qd_{vv}\right)^mi^{\lceil m\rceil-1}
\end{multline}
where we used Lemma \ref{L:exponentialfunction} and the fact, that $\displaystyle\rho_v<qd_{vv}/2$. Dividing both sides of (\ref{E:pr}) by $\displaystyle \left(\frac{q-1}{2}\right)^m$ and replacing $q-1$ by $q/2$ in denominator, we get
\begin{equation}\label{E:bd}
\sum_{w\in S^-_i}d_{ww}^m\leq 2\cdot 3^{m}\lceil m\rceil K^2d_{vv}^m i^{\lceil m\rceil-1}.
\end{equation}

On the other hand, from the triangle inequality and the definition of ring $R_i$, the following holds: $d_{wv}\geq d(s_w,s_v)-d(s_v,r_v)\geq (q-1)d_{vv}i$, so
\begin{equation}\label{E:ba}
a_{S^-_i}(v)\leq\frac{\sum_{w\in S^-_i}d_{ww}^\alpha}{\left((q-1)d_{vv}i\right)^\alpha}
\end{equation}
Using Lemma~\ref{L:ineq}, from~(\ref{E:bd}) and~(\ref{E:ba}) we get an upper bound on the affectance of the senders from $R_i$:
\begin{equation*}
a_{S^-_i}(v)<\frac{\left(\sum_{w\in S^-_i}d_{ww}^m\right)^{\alpha/m}}{((q-1)d_{vv}i)^\alpha}  \leq \frac{\left( 2\cdot 3^{m}\lceil m\rceil K^2d_{vv}^m i^{\lceil m\rceil-1}\right)^{\alpha/m}}{((q-1)d_{vv}i)^\alpha} \leq  \frac{3^{\alpha}\left(2\lceil m\rceil K^2\right)^{\alpha/m}}{(q-1)^\alpha i^{\alpha\left(\frac{m+1-\lceil m\rceil}{m}\right)}},i=1,2,\dots
\end{equation*}
By summing over $i$, and using (\ref{E:series}), we complete the proof of the lemma (as we have $\alpha>\displaystyle\frac{m}{m+1-\lceil m\rceil}$):
\begin{equation*}
a_{S^-}(v)\leq \frac{3^\alpha\left(2\lceil m\rceil K^2\right)^{\alpha/m}}{(q-1)^\alpha}\sum_{i=1}^{\infty}\frac{1}{i^{\alpha\left(\frac{m+1-\lceil m\rceil}{m}\right)}}\leq\frac{3^\alpha\left(2\lceil m\rceil K^2\right)^{\alpha/m}}{(q-1)^\alpha}\cdot\frac{\alpha\left(m+1-\lceil m\rceil\right)}{\alpha\left(m+1-\lceil m\rceil\right)-m},
\end{equation*}
so we have $\displaystyle c_0=\frac{3^\alpha\left(2\lceil m\rceil K^2\right)^{\alpha/m}\alpha\left(m+1-\lceil m\rceil\right)}{\alpha\left(m+1-\lceil m\rceil\right)-m}$ \qed
\end{proof}

Having Lemma~\ref{L:spl}, the proof of the following theorem is easy.
\begin{theorem}\label{T:feasibility}
If $c\geq \sqrt[\alpha]{\beta\left(c_0+1\right)}+3$ and $\alpha>\displaystyle\frac{m}{m+1-\lceil m\rceil}$, then the output of the algorithm is a feasible schedule.
\end{theorem}

	\subsection{The approximation ratio}
In this section we show that the algorithm outputs a schedule, which is longer than the optimal one no more than by a constant factor.

The following definition is taken from~\cite{fng}.
\begin{definition}
Let $S$ be a set of transmission requests and $p$ a node in the network, then we define 
\[
I_p(S)=\sum_{w\in S}\min{\left\{1,\left(\frac{d_{ww}}{d(s_w,p)}\right)^\alpha\right\}}\mbox{, and }I(S)=\max_{p}{I_p(S)}.
\]
\end{definition}
When $S$ is the set of all links (which we denoted by $L$), we use the notation $I(L)=I$. $I$ is a measure of interference, which in~\cite{fng} is shown to be a lower bound(with a constant factor) for optimal schedule length in case of linear power assignments.
\begin{theorem}~\cite{fng}\label{T:lowerbound}
If $T$ is the minimum schedule length, then $T=\Omega(I)$. 
\end{theorem}

Using Theorem~\ref{T:lowerbound} it is easy to prove the approximation ratio.
\begin{theorem}\label{T:basic}
a)~If $c>1$ and the output of Algorithm~\ref{A:LINEAR} is a feasible schedule, then it is a constant factor approximation for scheduling with linear powers,
b)~the optimal schedule length in case of linear powers is $\Theta (I)$.
\end{theorem}
\begin{proof}
Suppose $A_1,A_2,\dots,A_t$ is the output of Algorithm~\ref{A:LINEAR}. Let $v$ be a link from $A_t$. By definition of the algorithm we have 
\[
\displaystyle a_{A_i}(v)>\frac{1}{c^\alpha}\mbox{, if }i<t.
\]
 Since we assume $c>1$, we have also
 \[
 \displaystyle I_{r_v}(A_i)>\frac{1}{c^\alpha}\mbox{, so }\displaystyle I_{r_v}(L)=\sum_{i=1}^{t-1}I_{r_v}(A_i)>(t-1)/c^\alpha.
 \]
  On the other hand we have $I_{r_v}(L)\leq I(L)$, so
  \[
  t<c^\alpha I(L)+1\mbox{,}
  \]
 which together with Theorem~(\ref{T:lowerbound}) completes the proof.\qed
\end{proof}

\section{On the complexity of scheduling with linear powers}\label{complexity}
We define the problem EQSCHEDULING, which is a simplified case of the problem of scheduling (again w.r.t. linear power assignment), when all links in the network have ``almost the same'' length (i.e. the lengths differ only by a constant factor).

\textit{EQSCHEDULING:} given a set of links $L=\{l_1,l_2,\dots,l_n\}$ in a network, which have lengths differing not more than a constant factor, and a natural number $K>0$, the question is if there is a partition of that set into not more than $K$ SINR-feasible subsets or \emph{slots}.

 To show that the problem is NP-hard, we reduce to it the NP-complete problem PARTITION(see~\cite{gj}), which is defined as follows.

\textit{PARTITION:} given a finite set $A$ of positive integers, the question is if there is a subset $A'\subseteq A$, for which $\displaystyle \sum_{a\in A'}{a}=\sum_{a\in A\setminus A'}{a}$ holds, i.e. it is exactly the half of $A$ by sum. 

When it is not ambiguous, we will identify an instance of PARTITION with the corresponding set of integers $A$. 

Here we use not the general problem PARTITION, but some specific case, which is equivalent to the general one. We need the following lemma.

For a finite set of integers $A$ let  $S(A)$ denote the sum $\sum_{a\in A}{a}$.
\begin{lemma}\label{L:partition}
For each instance $A$ of PARTITION there is another instance $B$, which is polynomially equivalent to $A$, and satisfies the following properties:
	\begin{equation*} A\subset B\mbox{, and } \left| B\right|=3\left| A\right| \end{equation*}
	\begin{equation}\label{E:aina}
		 \mbox{for }a\in A, \frac{a}{S(B)}\leq \frac{1}{2\left| A\right|^3}
	\end{equation}
	\begin{equation}\label{E:ainb}
		 \mbox{for }a\in B\setminus A, \frac{a}{S(B)}\leq \frac{1}{2\left| A\right|}
	\end{equation}
\end{lemma}
\begin{proof}
Let $m$ denote a maximal element in $A$. Then we construct $B$ by adding $2\left|A\right|$ new elements with the same value $\left|A\right|^2m$ to $A$. As each element in $B\setminus A$ is not less than $S(A)$, then it's easy to check, that each partition of $A$ corresponds to a partition of $B$, and vice versa, so two instances of the problem are equivalent. The other two properties are straightforward.
\qed
\end{proof}
\begin{theorem}\label{T:complexity}
There are metric spaces, where  EQSCHEDULING is NP-hard.
\end{theorem}
\begin{proof}
We give a reduction from PARTITION. For a given instance $A$ of PARTITION let's construct the instance $B$, as described in Lemma \ref{L:partition}. We will construct an instance of EQSCHEDULING with a set of links $L$ in a network, so that the answer of $B$ is ``yes'' if and only if it is possible to schedule $L$ in two subsets. Let $B={a_1,a_2,\dots,a_n}$, and $n=|B|=3|A|$. The network has $2n+4$ nodes $s_i,r_i$, $i=0,1,\dots,n,n+1$. The set of links is $L=\{l_0,l_1,\dots,l_n, l_{n+1}\}$, where each link $l_i$ represents the sender-receiver pair $\left(s_i,r_i\right)$. The distances are defined as follows: we use $d_{ij}$ for denoting the distance $d(s_i,r_j)$. The links $l_0$ and $l_{n+1}$ have lenght $\displaystyle d_{00}=d_{n+1,n+1}=\frac{1}{\sqrt[\alpha]{3}}$. All other links have length 1. We set $d(r_0,r_{n+1})=0$, so that the links $l_0$ and $l_{n+1}$ cannot be scheduled together in the same group. Further,
\begin{gather*}
\displaystyle d_{i0}=d_{i,n+1}=\sqrt[\alpha]{\frac{\beta S(B)}{2  a_i}},\\
\displaystyle d_{0i}=d_{n+1,i} = d_{ii}+d_{i0}=\sqrt[\alpha]{\frac{\beta S(B)}{2 a_i}}+1,\\
\displaystyle d_{ij}=d(s_i,r_j)=1+d(s_i,r_0)+d(r_0, s_j)=1+\sqrt[\alpha]{\frac{\beta S(B)}{2 a_i}}+\sqrt[\alpha]{\frac{\beta S(B)}{2 a_j}}
\end{gather*}
for $i,j=1,2,\dots,n, i\neq j$. The other distances can be arbitrary, satisfying the axioms of metric spaces. First we show that the set of links $S=\{l_1,l_2,\dots,l_n\}$ is SINR-feasible. To do so we select any link, say the link $l_1$, and show, that the constraint (\ref{E:SINR}) is satisfied for $S$ and $l_1$. Taking into account the definition of the distances, the left part of (\ref{E:SINR}) becomes
\begin{multline*}
a_S(l_1)=\frac{1}{\displaystyle3\left(\sqrt[\alpha]{\frac{\beta S(B)}{2 a_1}}+1\right)^\alpha} + \sum_{i=2}^n{\frac{1}{\displaystyle\left(1+\sqrt[\alpha]{\frac{\beta S(B)}{2 a_i}}+\sqrt[\alpha]{\frac{\beta S(B)}{2 a_1}}\right)^\alpha}}\leq\\
 \leq\frac{1}{\displaystyle \frac{3\beta S(B)}{2 a_1}+3}+ \sum_{i=2}^n{\frac{1}{\displaystyle 1+\frac{\beta S(B)}{2 a_i}+\frac{\beta S(B)}{2 a_1}}}= \frac{1}{\beta}\left(\frac{1}{\displaystyle \frac{3S(B)}{2a_1}+\frac{3}{\beta}} + \sum_{i=2}^n{\frac{1}{\displaystyle\frac{1}{\beta}+\frac{S(B)}{2a_i}+\frac{S(B)}{2a_1}}} \right),
\end{multline*}
where we used the Lemma~\ref{L:ineq} with $s=\alpha$ and $r=1$, as we assume $\alpha>1$. For evaluating the last expression, we consider two cases:

1) if $1\in A$, then according to (\ref{E:aina}) we have $\displaystyle\frac{a_1}{S(B)}<\frac{1}{2|A|^3}$, so
\begin{multline*}
a_S(l_1)\leq \frac{1}{\beta}\left(\displaystyle \frac{1}{3|A|^3+3/\beta} + \sum_{i\in B\setminus A}{\frac{1}{\displaystyle 1/\beta+\frac{S(B)}{2a_i}+|A|^3}}+\sum_{i\in A\setminus \{1\}}{\frac{1}{\displaystyle 1/\beta+\frac{S(B)}{2a_i}+|A|^3}} \right)\leq\\
\leq \frac{1}{\beta}\left(\displaystyle \frac{1}{3|A|^3+3/\beta} + \frac{2|A|}{|A|+|A|^3+1/\beta}+\frac{|A|-1}{2|A|^3+1/\beta} \right).
\end{multline*}
As it is easy to see,  the right side is less than $1/\beta$ for instances with $|A|$ large enough.

2) if $1\notin A$, then according to (\ref{E:ainb}) we have $\displaystyle\frac{a_1}{S(B)}<\frac{1}{2|A|}$, so
\begin{multline*}
a_S(l_1)\leq \frac{1}{\beta}\left(\displaystyle \frac{1}{3|A|+3/\beta} + \sum_{i\in B\setminus \left( A\cup\{1\}\right)}{\frac{1}{\displaystyle 1/\beta+\frac{S(B)}{2a_i}+|A|}}+\sum_{i\in A}{\frac{1}{\displaystyle 1/\beta+\frac{S(B)}{2a_i}+|A|}} \right) \leq\\
\leq \frac{1}{\beta}\left(\displaystyle \frac{1}{3|A|+3/\beta} + \frac{2|A|-1}{2|A|+1/\beta}+\frac{|A|}{|A|^3+|A|+1/\beta} \right)=\\
=\frac{1}{\beta}\left(\displaystyle 1-\frac{|A|+3/\beta|A|+3/\beta^2+2/\beta}{6|A|^2+9/\beta|A|+3/\beta^2}+\frac{|A|}{|A|^3+|A|+1/\beta} \right).
\end{multline*}
The expression in parentheses is less than 1 when $|A|$ is large enough, so again $a_S(l_1)\leq 1/\beta$ holds.
 The affectance on the link $l_0$ by the set $S$ is equal to
 \[
  \displaystyle a_S(l_0)=\sum_{i=1}^n{\left(\frac{d_{ii}}{d_{i0}}\right)^\alpha}=\sum_{i=1}^n{\left(\frac{1}{\displaystyle\sqrt[\alpha]{\frac{\beta S(B)}{2  a_i}}}\right)^\alpha}=2/\beta,
  \]
  which is the same as the affectance on the link $l_{n+1}$ by $S$. So it follows, that $L$ can be scheduled into two sets if and only if the set $S$ can be partitioned into two subsets, so that each one affects the link $l_0$(and $l_{n+1}$) exactly by $\beta$. As it's not hard to check, this is the same as to solve the PARTITION problem instance $B$.
 The reduction is polynomial.
\qed
\end{proof}
\begin{corollary}
There are metric spaces, where EQSCHEDULING cannot be approximated within a constant factor less than $\displaystyle \frac{3}{2}$ unless $P=NP$.
\end{corollary}
\begin{proof}
Let $P\neq NP$. From the proof of Theorem~\ref{T:complexity} we see, that PARTITION can be polynomially reduced to EQSCHEDULING with $K=2$ in some metric spaces. Suppose there is a polynomial algorithm, which approximates EQSCHEDULING within a factor $\displaystyle\gamma<\frac{3}{2}$. Then let's consider an arbitrary instance $A$ (sufficiently large) of PARTITION. There is an instance $L$ of EQSCHEDULING with $K=2$, so that the answer for $A$ is 'yes' if and only if the linkset of $L$ can be scheduled in no more than 2 groups. Applying the approximation algorithm to $L$, we get the optimal schedule length with error at most a factor $\gamma$. If the resulting schedule has complexity not less than 3, then the optimal schedule length is at least $\displaystyle \frac{3}{\gamma}>2$, so the the answer of $A$ is 'no'. If the length of the resulting schedule is less than 3, then the optimal schedule length is no more than 2, so the answer of $A$ is 'yes'. This shows that the $\gamma$-approximation algorithm for EQSCHEDULING could be used to polynomially solve PARTITION, which is a contradiction to the assumption that $P\neq NP$.
\qed
\end{proof}

\end{document}